\documentclass[a4paper,UKenglish,cleveref, autoref, thm-restate]{lipics-v2021}

\hideLIPIcs %

\usepackage[rightcaption]{sidecap}

\usepackage{wrapfig}

\usepackage{hyperref}
\usepackage{titletoc}
\usepackage{blindtext}

\usepackage{algorithm}        %
\usepackage{algpseudocode}    %
\usepackage{eqparbox}         %
\usepackage{float}            %
\usepackage{multicol}

\usepackage{cleveref}
\usepackage{enumerate}

\usepackage{soul}  %

\usepackage{amsmath, amssymb, amsthm}
\usepackage{listings} %
\usepackage{xcolor} %
\usepackage{graphicx} %

\nolinenumbers

\newtheorem{invariant}[theorem]{Invariant}

\definecolor{codegreen}{rgb}{0,0.6,0}
\definecolor{codegray}{rgb}{0.5,0.5,0.5}
\definecolor{codepurple}{rgb}{0.58,0,0.82}
\definecolor{backcolour}{rgb}{0.95,0.95,0.92}

\algrenewcommand\algorithmiccomment[1]{\hfill // #1}

\newcommand{\CAS}{\op{CAS}}

\newcommand{\MIN}{\mbox{Min}}
\newcommand{\MAX}{\mbox{Max}}
\newcommand{\f}[1]{\mbox{\textit{#1}}} %
\newcommand{\x}[1]{\mbox{\textit{#1}}} %

\newcommand{\algorithmictype}{\textbf{class}}
\algdef{SE}[TYPE]{Type}{EndType}[1]{\algorithmictype\ #1}

\newcommand{\op}[1]{{\sf #1}} %
\newcommand{\boldop}[1]{\mbox{\textbf{\textsf{#1}}}}

\renewcommand{\algorithmiccomment}[1]{\hfill\eqparbox{COMMENT}{$\triangleright$\ #1}}
\newcommand{\Linecomment}[1]{$\triangleright$\ #1}

\titlecontents{section}
  [0pt] %
  {\addvspace{2pt}} %
  {\contentslabel{2em}} %
  {\hspace*{-2em}} %
  {\titlerule*[0.5pc]{.}\contentspage} %

\title{Concurrent Double-Ended Priority Queues\footnote{This research was supported by the Greek Ministry of Education, Religious Affairs and Sports though the HARSH project (project no. Y$\Pi$3TA - 0560901), which is carried out within the framework
of the National Recovery and Resilience Plan - Greece 2.0 - with funding from the European Union - NextGenerationEU.}}

\author{Panagiota Fatourou}{FORTH ICS, Greece \and University of Crete, Greece}{faturu@ics.forth.gr}{https://orcid.org/0000-0002-6265-6895}{}
\author{Eric Ruppert}{York University, Canada}{ruppert@eecs.yorku.ca}{https://orcid.org/0000-0001-5613-8701}{}
\author{Ioannis Xiradakis}{FORTH ICS, Greece \and University of Crete, Greece}{giannisx@ics.forth.gr}{https://orcid.org/0009-0009-4974-4901}{}

\authorrunning{P. Fatourou, E. Ruppert, I. Xiradakis}



\begin{document}

\maketitle

\begin{abstract}
This work provides the first concurrent implementation specifically designed
for a double-ended priority queue (DEPQ).
We do this by describing a general way to add an \op{ExtractMax} operation
to any concurrent priority queue that already supports \op{Insert} and \op{ExtractMin}  operations.
The construction uses two linearizable single-consumer priority queues to build a
linearizable dual-consumer DEPQ (only one process can perform \op{Extract} operations at each end).
This construction preserves lock-freedom.
We then describe how to use a lock-based combining scheme to allow multiple consumers
at each end of the DEPQ.
To illustrate the technique, we apply it to a list-based
priority queue.
\end{abstract}

\section{Introduction}
\label{introduction}

Priority queues, which store a set of keys and support an \op{Insert} operation
that adds a key to the set 
and an \op{ExtractMin} operation
that removes and returns the minimum key, have long been recognized 
as an important data structure for concurrent systems.
They have been used in operating systems for job queues and load balancing, 
heuristic searches \cite{RK88},
graph algorithms (e.g., \cite{BGM96,HW91}) and
for event-driven simulations \cite{Gon76,Jon86}.  
There are numerous concurrent implementations of priority queues, including  \cite{BCP16,HIST10,HW17,Joh94,LJ13,LS12,Low18,RT21,SL00,ST05,TMR15}.

In the single-process setting, there has also been much research on the problem of designing a \emph{double-ended} priority queue (DEPQ),
which supports an \op{ExtractMax} operation to extract the maximum key, as well
as the operation that extracts the minimum key.
Applications of DEPQs include external sorting \cite{handbook-depq}, 
coding theory \cite{SSF11},
and constraint programming~\cite{Cymer}. 
The problem of implementing a concurrent DEPQ,
has not been explored previously.
We approach this problem by giving a general transformation to
construct a concurrent DEPQ from a concurrent implementation of
an ordinary (single-ended) priority queue.
We assume a shared-memory system with asynchronous processes and we
use linearizability \cite{HW90} as the correctness condition for data structure
implementations.

Many of the single-process DEPQ implementations are based on
more complex variants of heaps like
min-max pair heaps \cite{OOW91}, deaps \cite{Car87}, min-max heaps \cite{ASSS86}, or others surveyed by Sahni~\cite{handbook-depq}.
Designing a linearizable concurrent version of such a data structure directly
would be quite challenging because an insertion or extraction
typically requires making multiple changes to the data structure.
In a concurrent setting, it is difficult to ensure that all the changes
appear to take place atomically.
Instead, we take our inspiration from 
the insight that some of these sequential DEPQ data structures
can be viewed as being constructed from a pair of single-ended priority queues \cite{CS00}.
Even using this approach, an update operation on the DEPQ often requires
two operations on the underlying priority queues to appear atomic.
Our construction uses a lightweight synchronization mechanism to 
achieve this.
Designing simple mechanisms to build complex concurrent
data structures in a modular way
is an important goal and can be quite challenging.
Indeed, in spite of the simplicity of our basic construction, the 
argument that it produces a linearizable DEPQ is subtle.

Our general construction in \Cref{generic} shows how to use two linearizable,
concurrent priority queues to construct a linearizable \emph{dual-consumer} DEPQ,
which allows only one process to perform \op{ExtractMin} and one process
to perform \op{ExtractMax} at a time.
The construction works even if the priority queues used are
\emph{single-consumer}, meaning that 
only one process at a time may perform \op{ExtractMin} operations.
Our DEPQ relies on the synchronization mechanisms of the 
single-ended priority queues to coordinate the actions of insertions and extractions
within each priority queue, but we add a new layer of synchronization between
the two processes that perform \op{ExtractMin} and \op{ExtractMax} operations.
Our construction preserves the property of linearizability.
It also preserves the lock-free progress property:  if the priority queue is
lock-free, then so is the resulting DEPQ.
Some concurrent priority queue data structures (for example, those based on concurrent skip lists) could be augmented to support deletion of an arbitrary element, given a reference to the location of the element in the data structure.
In this case, we prove
time and space bounds for our DEPQ, depending on the bounds for the (single-ended)
priority queue.

We describe how to adapt our dual-consumer DEPQ
to produce a linearizable DEPQ that can
handle multiple concurrent \op{Extract} operations at each end in \Cref{combining}.
A simple solution is to add two locks:  one that must be acquired
by \op{ExtractMin} operations and one for \op{ExtractMax} operations.
To reduce the overhead of having each \op{Extract} operation 
acquire and release a lock, we show how to use the combining technique \cite{FK12,HIST10}.
This is particularly appropriate for \op{Extract} operations on a priority
queue because all \op{Extract} operations attempt to remove an item from the
same location in the data structure, causing a hotspot of contention \cite{FKK22,HIST10}.
At a high level, each process that acquires the lock for one end of the DEPQ
can perform a whole batch of \op{Extract} operations,
lowering the overhead for synchronization.
Insertions can proceed concurrently with one another and with \op{Extract} operations.

Our technique is general enough to be applied to any concurrent priority queue.
For ease of illustration, we provide an example in \Cref{example}
that applies our technique (including combining) to a simple
list-based priority queue.
Since a single-consumer priority queue suffices for our construction,
this enables us to optimize the list-based priority queue using the simplifying assumption that only
one process can perform \op{ExtractMin} operations on it at a time.
We also use this example to illustrate how the code generated by 
our generic construction can be simplified when it is applied to 
a particular priority queue implementation.
For example, we save space by sharing nodes between the two lists
that represent the two priority queues.

We believe techniques similar to the one introduced here
could be applicable to other settings where a new
data structure can be built by combining two simpler data structures.
For example,
a binary search tree threaded with additional pointers to predecessors and successors of each node %
can be thought of as a binary search
tree and a doubly-linked list that share the same nodes.
Building concurrent implementations of the two components separately
is easier than trying to design one that maintains consistency among
all pointers of both structures simultaneously.

\section{Related Work}
\label{related}

\subsection{Sequential DEPQs}
\label{related-seq}
There is a large literature on building DEPQs that can be accessed by just a single process.
Many of them build
a DEPQ from two ordinary priority queues:  \x{MinPQ} is organized
for efficient \op{ExtractMin} operations and \x{MaxPQ}
for efficient \op{ExtractMax}
operations.
Chong and Sahni \cite{CS00} surveyed these results, classifying them into three types.

In the \emph{dual queue} method both \x{MinPQ} and \x{MaxPQ} contain all the keys
of the DEPQ.   
It is assumed that the priority queues also support the deletion of an arbitrary
key.  Maintaining pointers between the two copies of the key in the two
priority queues allows an extraction on one to also delete the item in the other queue.

The \emph{total correspondence} method stores half the DEPQ's items in \x{MinPQ} and 
the other half in \x{MaxPQ} and maintains the property that there is
a one-to-one correspondence between the keys in the two queues so that each
key in \x{MinPQ} is less than or equal to its
corresponding key in \x{MaxPQ}.  
If the DEPQ has an odd number of keys one
key is stored in a special buffer.  Using pointers between pairs of
corresponding keys, a few keys are reshuffled using the buffer
to maintain the correspondence
when a key is inserted into or deleted from the DEPQ.
See \cite{CS99} for a discussion of several DEPQs of this form.

The third, less-general method, \emph{leaf correspondence}, is like total correspondence,
but only a subset of the items are included in the one-to-one correspondence between keys.

Another technique \cite{EJK08}
is to periodically rebuild a data structure so that the larger half
of the items are stored in \x{MaxPQ} and the smaller half in \x{MinPQ}.
When one of the two priority queues becomes empty, the data structure
is rebuilt, redistributing the keys.
This approach has good amortized complexity, but can also be deamortized.

\subsection{Concurrent Data Structures}

The sequential methods for building a DEPQ from standard priority queues described above
require making several updates to the priority queues to
implement an operation on the DEPQ.  Thus, it may be challenging
to use them in a concurrent setting because these multiple changes
cannot be done atomically.  In this paper, we describe one
way to adapt the dual queue method to obtain a concurrent DEPQ.
This requires a new mechanism for 
\op{ExtractMin} and \op{ExtractMax} operations to coordinate
with one another.
We do not assume the existence of a \op{Delete} operation on the basic priority queues.

Any linearizable concurrent priority queue can be used in our construction.
See \cite{DB08,MS18-48.8.1} for a survey of early work on concurrent priority queues,
including many based on sequential heap data structures for priority queues.
Tamir, Morrison and Rinetzky \cite{TMR15} added support for an operation that 
modifies an existing key in a lock-based concurrent heap.
Lowe \cite{Low18} designed a lock-free binomial heap, which includes support
for merging two priority queues.
Pugh's skip list data structure \cite{Pug90} has been used as the basis for 
several concurrent priority queues \cite{BCP16,LJ13,SL00,ST05}.
Liu and Spear \cite{LS12} gave a novel concurrent priority queue data structure based on a tree where each node stores a sorted linked list.

The concurrent priority queues discussed above allow multiple concurrent consumers.
For our construction, it suffices to have a single-consumer priority queue,
which may be possible to implement more efficiently.
Hoover and Wei \cite{HW17} give a wait-free implementation of a single-consumer priority queue that has $O(\log n + \log p)$ step complexity per operation,
where $n$ is the number of elements in the queue and $p$ is the number
of processes accessing the queue.

We are unaware of any previous work that aims to provide a linearizable concurrent DEPQ implementation.
Medidi and Deo \cite{MD94} described performing batches of operations on a DEPQ in parallel, but their focus
was on the synchronous PRAM model.

One way to use existing data structures to build a linearizable concurrent DEPQ is to use a binary search tree (BST),
where the tree is sorted by key values \cite[Section 5.2.3]{TAOCP3}.
There are a number of lock-free concurrent BST implementations (e.g., \cite{EFRvB10,NRM20}) that support insertion and deletion of keys.
The \op{Delete} operation in these BSTs could easily be modified to delete (and return)
the minimum or maximum key present in the BST to yield the \op{Extract} operations of a DEPQ.
Since repeated \op{Extract} operations at one end of the DEPQ could yield a lopsided
BST, it would likely be desirable to use a balanced BST, such as the concurrent
chromatic tree of \cite{BER14}.
However, even with balancing, this would require each \op{Extract} operation to
traverse a path of length $\Theta(\log n)$ in the BST when the DEPQ contains $n$ elements.
In contrast, \op{Extract} operations in the list-based DEPQ described in \Cref{example} finds the required key right at the beginning of the list, and requires
less restructuring compared to the rebalancing of chromatic trees.
There are also concurrent implementations of (single-ended) priority queues that augment a search tree with a sorted linked list of keys in the tree to expedite \op{ExtractMin} operations \cite{Joh94,RT21}.

\subsection{Combining}
\label{combining-related}

\emph{Software combining} is a technique for reducing the overhead of using locks.
The process that acquires the lock performs its own work as well as the work of
other processes that failed to acquire the lock.
Early versions \cite{GVW89, YTL87} used a tree to combine requests for work.
Hendler et al.\ \cite{HIST10} developed a more efficient approach called flat combining.
Fatourou and Kallimanis \cite{FK12} improved the efficiency using an
integrated scheme for both locking the data structure and
organizing the requests for operations on it.
See \Cref{combining} for more details.

Hendler et al.\ \cite{HIST10}
illustrated their combining technique by applying combining to a skip list to obtain
a concurrent (single-ended) priority queue.
They combined both \op{Insert} and \op{ExtractMin} operations so that they
could use a sequential priority queue as the underlying data structure,
which severely restricts parallelism.
In our construction, we use combining only for the \op{ExtractMin} operations,
so \op{Insert} operations may run concurrently.

\section{Constructing a Dual-Consumer DEPQ from Two Priority Queues}
\label{generic}

In this section, we describe how to build a linearizable dual-consumer DEPQ from two single-consumer (single-ended) priority queues.
The priority queues used in the construction support \op{Insert} and \op{ExtractMin} operations, while the DEPQ supports \op{Insert}, \op{ExtractMin} and \op{ExtractMax}
operations.

Let $R$ be a total order on a universe Key of possible keys.
Let $PQ(R)$ be a linearizable implementation  of a priority queue
that supports \op{Insert} and \op{ExtractMin} operations, where the latter operation
removes and returns an element that is minimal according to the relation $R$.
Let $<$ be a total order on Key and let $>$ be the opposite total order.  (That is,
$x<y$ if and only if $y>x$.)
To build a linearizable DEPQ we use one instance of $PQ(<)$
called \x{MinPQ} and one instance of $PQ(>)$ called \x{MaxPQ}.
(Thus, an \op{ExtractMin} on \x{MaxPQ} returns the largest key with respect to the total order $<$.)
Pseudocode for our implementation is given in \Cref{generic-alg}.

An \boldop{Insert} operation on the DEPQ simply inserts the key into both priority queues.
To coordinate extractions, each item has an associated \x{reserved} bit, which is initially 0.
An \boldop{ExtractMin} on the DEPQ repeatedly extracts the minimum element from \x{MinPQ}
and tries to set its \x{reserved} bit using a \op{TestAndSet} instruction
until it successfully changes the \x{reserved} bit of some item.  The key of that item
is then returned.
If, at any time, the \op{ExtractMin} observes that \x{MinPQ} is empty,
it terminates and indicates that the DEPQ is empty.
The \boldop{ExtractMax} operation on the DEPQ is symmetric to \op{ExtractMin}, using \x{MaxPQ}
in place of \x{MinPQ}.

The \x{reserved} bit ensures that an element cannot be returned twice by both an
\op{ExtractMin} and an \op{ExtractMax}.
An element removed from the DEPQ by an \op{ExtractMin} operation may remain in \x{MaxPQ} for
some time, but if it is eventually removed from \x{MaxPQ} by an \op{ExtractMax} operation,
the \op{ExtractMax} will skip the item because its \op{TestAndSet} operation will fail to 
set the item's \x{reserved} bit.

If the priority queue implementation we are using also supports a \op{Delete} operation,
then we can add the optional lines \ref{mindel} and \ref{maxdel} to the \op{Extract} operations.
The DEPQ is linearizable regardless of whether these lines are included or not.
When an \op{ExtractMin} operation removes an item from \x{MinPQ}, line \ref{mindel}
removes it from \x{MaxPQ}.  It is not strictly necessary to remove it from \x{MaxPQ}
because its \x{reserved} bit is set, so any \op{ExtractMax} operation that finds
it in \x{MaxPQ} would ignore it anyway.
Including lines \ref{mindel} and \ref{maxdel} allows us to prove the complexity
bounds in \Cref{complexity}.  Whether their inclusion improves performance
may depend upon the underlying priority queue: 
if deleting an arbitrary element is not too much more costly than extracting
the minimum, or if performance of the priority queues would degrade significantly due to the presence
of obsolete items, then it may be worthwhile to include lines \ref{mindel} and \ref{maxdel}.

\begin{algorithm}[t]
\begin{algorithmic}[1]
\Type item
	\State Key \x{key}
	\State boolean \x{reserved} 
\EndType

\medskip

\Function{Insert}{Key $x$}
	\State \x{item} $i :=$ new \x{item} with $\x{key} = x$, $\x{reserved}=0$
	\State $\x{MinPQ}.\op{Insert}(i)$ \label{minins}
	\State $\x{MaxPQ}.\op{Insert}(i)$ \label{maxins}
\EndFunction

\medskip

\Function{ExtractMin}{} : Key
	\While{true}  \Comment{repeatedly extract element from \x{MinPQ}}\label{exmin-loop}
		\State $x := \x{MinPQ}.\op{ExtractMin}$ \label{minex}
		\If{$x=nil$} \Return $nil$\label{minnil}\Comment{DEPQ is empty}
	    \ElsIf{$\op{TestAndSet}(x.\x{reserved}) =0$} \Comment{return when \op{TestAndSet} succeeds}
	     	\State \x{MaxPQ}.\op{Delete}($x$) \Comment{this line is optional}\label{mindel}
			\State \Return $x.\x{key}$ \label{mints}
	    \EndIf
	\EndWhile 
\EndFunction

\medskip

\Function{ExtractMax}{} : Key
	\While{true} \Comment{repeatedly extract element from \x{MaxPQ}}\label{exmax-loop}
		\State $x := \x{MaxPQ}.\op{ExtractMin}$ \label{maxex}
		\If{$x=nil$} \Return $nil$\label{maxnil} \Comment{DEPQ is empty}
	    \ElsIf{$\op{TestAndSet}(x.\x{reserved}) =0$} \Comment{return when \op{TestAndSet} succeeds}
	    	\State \x{MinPQ}.\op{Delete}($x$) \Comment{this line is optional}\label{maxdel}
	    	\State \Return $x.\x{key}$ \label{maxts}
	    \EndIf
	\EndWhile 
\EndFunction
\end{algorithmic}
\caption{Generic construction of a DEPQ from two priority queues.\label{generic-alg}}
\end{algorithm}

\subsection{Linearizability}

There are challenges in showing that the DEPQ is linearizable.
Since \x{MinPQ} and \x{MaxPQ} are updated separately by \op{Insert} operations,
their contents may not exactly match.
Moreover, since the \x{reserved} bit is updated \emph{after} removing the item
from \x{MinPQ} or \x{MaxPQ}, the order in which items' bits are set may be different
from the order in which they are removed from the priority queues.

Since we have assumed that the implementations of \x{MinPQ} and \x{MaxPQ} are linearizable, the composability property
of linearizability~\cite{HW90} allows us to consider operations applied to each of these priority queues as atomic steps. 
Thus, when we refer to a step later on, this step might be the execution of (an entire) such operation. 
For the sake of simplicity in our presentation, we assume that all keys inserted into the DEPQ are distinct.
This allows us to talk about the \op{Insert} that inserted the key extracted
by some \op{ExtractMax} or \op{ExtractMin} operation without ambiguity.

Fix an execution $\alpha$. We now describe how to linearize operations in $\alpha$.
An \op{Extract} operation that never performs an \op{ExtractMin} at line \ref{minex}
or \ref{maxex} is not linearized, since it never accesses shared memory.
For each \op{Extract} operation $e$ that does perform an \op{ExtractMin} at line
\ref{minex} or \ref{maxex}, let $last_e$ be the last step where $e$ does so,
and let $result_e$ be the key returned by this last \op{ExtractMin} performed by $e$.

\begin{enumerate}[L1.]
\item\label{ex-lin1}
An \op{Extract} operation $e$ is linearized at $last_e$ if either 
$result_e=nil$ or $e$ performs a successful
\op{TestAndSet} at line \ref{mints} or \ref{maxts}.
\item
Linearize each \op{Insert} operation at the earlier of
\begin{enumerate}
\item\label{ins-lin1}
its insertion into \x{MaxPQ} at line \ref{maxins}, or
\item\label{ins-lin2}
immediately before the \op{ExtractMin} operation on DEPQ that returns its item.
\end{enumerate}
If neither of these events occur, then the \op{Insert} is not linearized.
\end{enumerate}

The following two observations ensure that the linearization respects the real-time
order of operations.

\begin{observation}\label{assigned}
Every operation that terminates is assigned a linearization point.
\end{observation}
\begin{proof}
If an \op{Extract} operation $e$ terminates, 
it satisfies one of the exit conditions at line \ref{minnil}, \ref{mints}, \ref{maxnil} or \ref{maxts}, so it is linearized according to L\ref{ex-lin1}.
If an \op{Insert} terminates,
it must execute line \ref{maxins}, so it is assigned a linearization point.
\end{proof}

\begin{observation}\label{real-time}
If an operation is linearized, its linearization point is in its execution interval.
\end{observation}
\begin{proof}
An \op{Insert} linearized according to L\ref{ins-lin2} is before it executes
line \ref{maxins} and after it executes line \ref{minins} (since some \op{Extract}
operation removes its element from the DEPQ), so the \op{Insert}
is linearized during its execution interval.
All other operations are linearized at one of their own steps.
\end{proof}

We are now ready to prove the main lemma, which shows that the results returned in the concurrent execution $\alpha$ are consistent with the linearization.
\begin{lemma}\label{right-responses}
In the sequential execution defined by the linearization,
each \op{Extract} operation $e$ returns $result_e$. 
\end{lemma}
\begin{proof}
We prove the claim by strong induction.
Consider any \op{Extract} operation $e$ in the linearization.
Let $L_e$ be the sequential execution defined by the prefix of the linearization of $\alpha$ that consists of all operations prior to $e$ in the linearization.
As our induction hypothesis, assume that all \op{Extract} operations in $L_e$ satisfy the claim.
We must show that $e$ returns $result_e$ in the sequential execution.

Let $C$ be the configuration just before the linearization point of $e$.
Let $Q$ be \x{MinPQ} if $e$ is an \op{ExtractMin} or \x{MaxPQ} if $e$ is an \op{ExtractMax}.
Our goal is to show that every key that is in the DEPQ at the end of $L_e$ is
in $Q$ at $C$.  We do this by proving the following two claims.

\begin{claim}\label{inserted}
All keys inserted into but not extracted from the DEPQ during $L_e$ are inserted into $Q$ before $C$.
\end{claim}

\begin{claimproof}
First, in the case where $e$ is an \op{ExtractMin}, we show that all keys inserted by operations in $L_e$ are inserted into \x{MinPQ}.
Consider any key $k$ for which there is an \op{Insert}($k$) in $L_e$.
If the \op{Insert} is linearized at line \ref{maxins} by rule L\ref{ins-lin1}, then $k$ has
been inserted into \x{MinPQ} at line \ref{minins} before $C$.
Otherwise, if the \op{Insert} is linearized by rule L\ref{ins-lin2},
then it has been removed from \x{MinPQ} by an \op{ExtractMin} before~$C$,
so it must have been added to \x{MinPQ} before $C$.
Thus, all keys inserted by operations in $L_e$ have been inserted into \x{MinPQ} before $C$.

Now, consider the case where $e$ is an \op{ExtractMax}.
We again argue that all keys inserted by operations in $L_e$ are inserted into \x{MaxPQ}. 
Every key whose insertion in $L_e$ is linearized according to L\ref{ins-lin1}
has been inserted into \x{MaxPQ} before $C$.
Now consider any key whose insertion in $L_e$ is linearized by L\ref{ins-lin2}.
Then this key has been extracted during $L_e$, so \Cref{inserted} holds trivially.
\end{claimproof}

\begin{claim}
The key of any item removed from $Q$ before $C$ is also removed by an \op{Extract} operation on the DEPQ in $L_e$.
\end{claim}
\begin{claimproof}
Let $k$ be the key of some item \x{item} removed from $Q$ before $C$.
We argue that some \op{Extract} operation performed a \op{TestAndSet}
on the \x{reserved} bit of \x{item} before $C$.
If $e$ itself did the \op{Extract} operation on $Q$ that removed \x{item} before $C$, then 
it must have done so in an iteration of the loop at line \ref{exmin-loop} or \ref{exmax-loop} prior to the iteration that
contains $e$'s linearization point of $e$, which is after $C$.
So, $e$ performed a \op{TestAndSet} on \x{item} in that prior iteration.
If another \op{Extract} operation on $Q$ removed \x{item},
that operation terminated before $e$ began (since there is only
one \op{Extract} operation of each type can be running at any time),
so it must have performed the \op{TestAndSet} on \x{item} before $e$ began.
Finally, if some operation performed a \op{Delete} on $Q$ that removed \x{item} (at line \ref{mindel} or \ref{maxdel}), then it performed a \op{TestAndSet} on \x{item}'s \x{reserved} bit on the previous line.
In all cases, some process performed a \op{TestAndSet} on \x{item}'s \x{reserved} bit before $C$.

Whichever \op{Extract} operation $e'$ performed the successful
\op{TestAndSet} on this bit
was linearized before $C$ and therefore is in $L_e$.
Moreover, $result_{e'}=k$ since $e'$ performs the successful \op{TestAndSet} on  \x{item}.\x{reserved} 
in its last iteration of the loop at line \ref{exmin-loop} or \ref{exmax-loop} (by the exit condition
on line \ref{mints} or \ref{maxts}), and $e'$ gets \x{item} from its \op{Extract} operation at line \ref{minex} or \ref{maxex} of that iteration.
So, by the induction hypothesis, $e'$ returns $k$ in $L_e$.
\end{claimproof}

It follows from the two claims that every key that is in the DEPQ at the end of $L_e$ is also in $Q$ at~$C$.
Recall that $result_e$ is the result of the \op{Extract} on $Q$ that $e$ performs
in the next step after $C$.
We now consider two cases to complete the induction step.

If $result_e=nil$ then the DEPQ is empty at the end of $L_e$,
so $e$ returns nil in the sequential execution defined by the linearization.

If $result_e$ is a non-\x{nil} key $k$, then we first show that an \op{Insert}($k$)
is in $L_e$.
If $e$ is an \op{ExtractMax}, then $e$ extracted $k$ from \x{MaxPQ}, so
its \op{Insert} was linearized before $C$ when $k$ was inserted into \x{MaxPQ}.
If $e$ is an \op{ExtractMin}, rule L\ref{ins-lin2} guarantees that the
\op{Insert}($k$) was linearized before $C$.
We next argue that $k$ has not been removed from the DEPQ
by any \op{Extract} operation in $L_e$.
If there were such an \op{Extract} operation $e'$ in $L_e$, it would have 
set the \x{reserved} bit of $k$ in $\alpha$, which is impossible since $e$ sets
the \x{reserved} bit of $k$ in $\alpha$.
So, $k$ is in the DEPQ at the end of $L_e$.
 Since every key in the DEPQ at the end of $L_e$ is in $Q$ at $C$,
 it follows that the result $k$ returned by $e$ in the concurrent execution is the returned by $e$ in the sequential execution defined by the linearization.
\end{proof}

Combining \Cref{assigned}, \Cref{real-time} and \Cref{right-responses}
yields the following result.

\begin{theorem}
If \x{MinPQ} and \x{MaxPQ} are linearizable (single-consumer) priority queues, then
\Cref{generic-alg} is a linearizable
dual-consumer DEPQ.
\end{theorem}

We remark that \Cref{generic-alg} is \emph{not} linearizable if multiple processes can perform
\op{Extract} operations at the same end of the DEPQ.
For example, consider the following non-linearizable execution.
\begin{itemize}
\item Two items 1 and 2 are inserted into both \x{MinPQ} and \x{MaxPQ}.
\item Then, an \op{ExtractMin} $E_1$ removes 1 from \x{MinPQ} and goes to sleep.
\item Next, a second \op{ExtractMin} $E_2$ removes 2 from \x{MinPQ}, performs a successful \op{TestAndSet} and returns 2.
\item Finally, an \op{ExtractMax} $E_3$ removes 2 from \x{MaxPQ} but fails the \op{TestAndSet}, so it removes 1 from \x{MaxPQ}, does a successful \op{TestAndSet} on that item and returns 1.
\end{itemize}
$E_2$ terminates before $E_3$ begins, so these operations must be linearized
in this order.  $E_1$ cannot be linearized before $E_3$ (since then $E_2$ and $E_3$ could not both return non-nil values).
In the sequential execution corresponding to the linearization, $E_2$ should return 1, which is not what it returns in the concurrent execution.

\subsection{Progress and Complexity Bounds}
\label{complexity}

The \op{Insert} operation on the DEPQ is lock- or wait-free if insertions
on the underlying priority queues are.
The \op{Extract} operations on the DEPQ is lock-free if extractions (and deletions,
if lines \ref{mindel} and \ref{maxdel} are included) on the underlying priority 
queues are.  This is because the \op{TestAndSet} on line \ref{mints} or \ref{maxts}
can fail only if some other \op{Extract} operation has successfully performed the
\op{TestAndSet}, and that other \op{Extract} operation is guaranteed to terminate.

\op{Extract} operations on the DEPQ are not wait-free, even 
if the underlying priority queue is wait-free because one \op{Extract}
operation may repeatedly fail its \op{TestAndSet} in every iteration.
Consider an initially empty DEPQ where one process repeatedly inserts
an item into \x{MinPQ} and \x{MaxPQ} and then removes it from \x{MinPQ} and successfully
performs a \op{TestAndSet} on its \x{reserved} bit and returns it.  
Meanwhile, another process performs an \op{ExtractMax} operation in which each
iteration removes an item from \x{MaxPQ} but fails its \op{TestAndSet} on 
its reserved bit.  The second process may continue forever.
Modifying the construction to make it wait-free would probably require some
coordination between the two processes performing \op{Extract} operations at opposite ends of the DEPQ.

Next, we prove  bounds on the space and step complexity of our
DEPQ.  They hold if lines \ref{mindel} and \ref{maxdel} are included
in the \op{Extract} routines.  
We provide amortized bounds on step complexity, which bound
the average time per operation in the worst-case execution if the underlying priority queue is lock-free.
More precisely, saying that the amortized step complexity is $O(T_i)$ for insertions and $O(T_e)$ for an \op{Extract} operation means that, for every (finite) execution in which $m_i$ \op{Inserts} and $m_e$ \op{Extracts} are invoked, the total number of steps in the execution is $O(m_i\cdot T_i + m_e \cdot T_e)$.

Let $n$ be the maximum number of elements in the data structure at any time, if operations are performed sequentially in the order of the linearization.
Let $c$ be the maximum point contention, that is, the maximum 
number of operations that are active at any one time.
Assume that the underlying priority queues have space complexity
$S(n,c)$ and that the (amortized) step complexity for their \op{Insert}, \op{ExtractMin} and \op{Delete} operations are
$T_i'(n,c), T_e'(n,c)$ and $T_d'(n,c)$, respectively.

The maximum number of elements that are ever in \x{MinPQ} or \x{MaxPQ}
is $O(n+c)$.
Thus, the amortized step complexity for \op{Insert} and \op{Extract} operations
on the DEPQ are $O(T_i'(n+c,c))$ and $O(T_e'(n+c,c) + T_d'(n+c,c))$, respectively.
The amortized bound for \op{Extract} operations follows from the fact that the cost
of each iteration of its loop that fails its \op{TestAndSet} can be charged
to the \op{Extract} operation that performs the successful \op{TestAndSet} on the same item.
Each \op{Extract} operation $e$ is charged for at most one iteration of at most one other \op{Extract},
so this contributes only $T_e'(n+c,c)+T_d'(n+c,c)+O(1)$ steps to $e$'s amortized step
complexity.

If we have bounds on the \emph{expected} step complexity of the underlying
data structure (for example, if the data structure is randomized, like a skip list)
then the bounds described above hold for the expected amortized step complexity of 
operations on the DEPQ.
If $T_i$ is a bound on the worst-case time per insertion into the underlying priority
queue, rather than an amortized bound, then the bound for insertions on the DEPQ
is also a worst-case bound.

\section{Constructing a Multi-Consumer DEPQ from a Dual-Consumer DEPQ}
\label{combining}

The construction provided in \Cref{generic} only permits a single \op{Extract}
operation at each end of the DEPQ at a time.
If multiple processes wish to perform \op{Extract} operations, we could
use two locks:  one for each end.
A process must acquire the
lock for one end of the DEPQ before performing an \op{Extract} on that end.
\op{Insert} operations can proceed regardless of the locks.

To reduce the overhead required by locking, we can use the combining technique mentioned in \Cref{combining-related},
wherein one process is selected at each point in time to perform the pending operations of a group of processes
waiting to acquire the lock.
We describe how to use the combining approach on our DEPQ with the 
combining algorithm of Fatourou and Kallimanis \cite{FK12}, called CC-Synch.

In CC-Synch, a list is employed to store the operations
of active processes. 
CC-Synch implements a lock, which is used by the processes to choose one process, 
called a \emph{combiner}, which is delegated to perform a batch of pending operations.
After announcing its requested operation, 
each process $p$ that is not chosen as the combiner simply waits
(by performing local spinning) until the
combiner informs $p$ that its requested operation has been completed and
provides the result of its operation.
Roughly speaking, CC-Synch implements
a combining-friendly variant of a queue lock~\cite{C93,MLH94,MCS91},
and the queue that is used to implement the lock is also used to store the
operations of the active processes in FIFO order.
This design yields better performance than other combining algorithms \cite{FK12}.

We use two instances of CC-Synch, one for each end of the DEPQ.
A process adds its request to do an \op{Extract} operation on the DEPQ to 
the FIFO queue of the appropriate instance of CC-Synch.
When CC-Synch chooses a combiner process to perform a batch of 
\op{Extract} operations from the FIFO queue, the combiner performs the \op{Extracts}  one by one
using \Cref{generic-alg}.  Since CC-Synch ensures that only one process
acts as a combiner at a time, this satisfies \Cref{generic-alg}'s requirement
that there be only one process performing \op{Extract} operations at each end of the DEPQ.
This approach
improves locality in accessing the data structure, since the combiner process 
performs an entire batch of \op{Extract} operations.  
It also avoids having a hotspot of contention because multiple 
\op{Extract} operations would typically access the same part of the data structure.
We remark that there is little contention on the \x{reserved} bits of items:
only the two \op{Extract} operations that extract an item from \x{MinPQ} and \x{MaxPQ}
can ever access its \x{reserved} field, and the \x{reserved} field is ignored
by \x{Insert} operations.

\section{Example: a List-Based DEPQ}
\label{example}

\newcommand{\nt}{\mbox{$\x{next}[\x{type}]$}}
\newcommand{\head}{\mbox{$\x{Head}[\x{type}]$}}
\newcommand{\ld}{\mbox{$\x{LastDeleted}[\x{type}]$}}

State-of-the-art concurrent priority queues are based on skip lists \cite{LJ13,ST05}.
They  maintain a sorted singly-linked list of keys in the set 
with additional levels of pointers to expedite accesses to the list.
These concurrent priority queues can be used as a black box in our construction
to obtain a DEPQ.
However, most concurrent priority queues are quite complex, so 
we use a simpler example to illustrate our construction.
This example uses a singly-linked list as the basic priority queue, 
combining ideas from Lind\'{e}n
and Jonsson's priority queue  based on a skip list \cite{LJ13} with
the combiner of Fatourou and Kallimanis \cite{FK12}.
This allows us to describe several optimizations that can be made when 
using priority queues based on lists (or skip lists).
The resulting DEPQ described in this section is the first linearizable
concurrent data structure designed to implement a DEPQ that we are aware of.

\begin{algorithm}[t]
\begin{algorithmic}
\State constant int \MIN\ := 0
\State constant int \MAX\ := 1
\State

\Type Node
	\State Key \x{key}
	\State boolean \x{reserved}
	\State array[\MIN..\MAX] of $\langle \mbox{Node*, boolean}\rangle$ \x{next} \Comment{pointer and bit are stored in a single word}
\EndType
\Type DEPQ
	\State array [\MIN..\MAX] of Node* \x{Head}
	\State array [\MIN..\MAX] of Node* \x{LastDeleted}
	\State CC-Synch $C_{min}, C_{max}$
\EndType
\end{algorithmic}
\caption{Object types used by list-based DEPQ.\label{list-type}}
\end{algorithm}

Our example DEPQ applies the construction of \Cref{generic-alg}
to two priority queues \x{MinPQ} and \x{MaxPQ}, built from singly-linked lists $L_{min}$ and $L_{max}$, respectively.
Instead of having two totally separate singly-linked lists, we can save space by sharing nodes between the two lists.
Thus, each element added to the DEPQ is represented by a node
with fields to store the key, pointers to the next node in
each of $L_{min}$ and $L_{max}$, and the \x{reserved} bit.
(See \Cref{list-type}.)

We first give a high-level description of our list-based priority queue.
An \op{Extract} operation deletes a node in two phases.
The node is first logically deleted by marking the \x{next} field of 
the \emph{previous} node.
The logically deleted nodes form a prefix of the list,
and the marking mechanism ensures that no new keys are inserted into this prefix.
Periodically, a chain of logically deleted nodes is physically deleted 
from the front of the list
by advancing the pointer that stores the head of the list.

The nodes beyond the logically deleted prefix are kept sorted by their keys
in increasing order in $L_{min}$ and in decreasing order in $L_{max}$.
This ensures that each \op{Extract} operation removes the correct key.
An \op{Insert} operation traverses the list, passing the logically deleted
prefix until finding the correct location to insert its new key,
and then attempts to insert it with a CAS instruction.
If the CAS fails, the \op{Insert} tries again by traversing forward
from the location where it failed.

Since \op{ExtractMin} operations naturally involve a hotspot of contention
at the front of the list, we follow the approach of 
Lind\'{e}n and Jonsson \cite{LJ13} to make
them lightweight.  The logical deletion by an \op{ExtractMin} requires
only a single successful read-modify-write instruction (which can be a CAS
instruction or a \op{FetchAndOr}, if available).
The physical deletion is performed periodically by a separate \op{SL-UpdateHead} routine that advances the head pointer to the last logically deleted node.

As described in \Cref{combining}, we use one instance of the CC-Synch
algorithm to coordinate \op{Extract} operations on each end of the DEPQ.
We add the following optimization.
To perform a batch of \op{Extract} operations, the combiner first
logically deletes a sequence of nodes at the beginning of the
list.  Then, it updates the head pointer just once to physically remove all of them.

\subsection{Detailed Description}

\begin{algorithm}
\caption{Pseudocode for list-based DEPQ.\label{list-code}}
\begin{algorithmic}[1]
\setcounter{ALG@line}{99}
{\small
\Function{\op{Insert}}{Key $k$}\Comment{insert $k$ into DEPQ}
    \State Node* \x{node} := \textbf{new} Node with \x{key} $k$ and \x{reserved} bit 0\label{L:newNode}
    \State \op{SL-Insert}(\x{node}, \MIN)\label{L:insMin} \Comment{insert into singly-linked list $L_{min}$}
    \State \op{SL-Insert}(\x{node}, \MAX)\label{L:insMax} \Comment{insert into singly-linked list $L_{max}$}
\EndFunction

\medskip

\Function{\op{ExtractMin}}{} : Key \Comment{remove and return min key from DEPQ}
    \State \Return $C_{min}$.\op{CCSynch}(\op{SL-Extract}(\MIN), \op{SL-UpdateHead}(\MIN)) \label{L:DeleteMin}
\EndFunction

\medskip

\Function{\op{ExtractMax}}{} : Key\Comment{remove and return max key from DEPQ}
    \State\Return $C_{max}$.\op{CCSynch}(\op{SL-Extract}(\MAX), \op{SL-UpdateHead}(\MAX)) \label{L:DeleteMax}
\EndFunction

\medskip

\Function{\op{SL-Insert}}{Node* \x{node}, int \x{type}} \Comment{insert $node$ into $L_{type}$}
	\State Node* \x{pred} := \x{Head}[\x{type}]\label{init-pred}
	\State Node* \x{curr}
	\State boolean \x{marked}
	\Repeat \Comment{repeatedly try inserting node}\label{L:ins0-repeat}
		\State $\langle \x{curr}, \x{marked}\rangle :=  \x{pred}\rightarrow \x{next}[\x{type}]$\label{L:readcurr1}
		\While{$\op{Before}(\x{curr}\rightarrow key, k, \x{type}) \mbox{ or } \x{marked}$}\label{L:ins-while}\Comment{advance to correct location}
			\State $\x{pred} := \x{curr}$
			\State $\langle \x{curr},\x{marked}\rangle := curr\rightarrow \x{next}$\label{L:readcurr2}
		\EndWhile\label{L:ins-endwhile}
	 	\State $\x{node} \rightarrow \x{next}[\x{type}] := \langle \x{curr},0\rangle$\Comment{set \x{next} field of new node}\label{L:ins0-newSucc}
	\Until \op{CAS}($\&\x{pred} \rightarrow \x{next}[\x{type}], \langle \x{curr}, 0\rangle, \langle \x{node}, 0\rangle$)\Comment{try  to insert \x{node} between \x{pred} and \x{curr}}\label{L:ins0-cas}
\EndFunction

\medskip

\Function{\op{Before}}{Key $k_1$, Key $k_2$, int \x{type}} : boolean \Comment{should $k_1$ come before $k_2$ in list $L_{type}$?}
	\If{$\x{type} = \MIN$} \Return $k_1<k_2$
	\Else\ \Return $k_2 < k_1$
	\EndIf
\EndFunction

\medskip

\Function{\op{SL-Extract}}{int \x{type}} : Key \Comment{\op{ExtractMin} or \op{ExtractMax}, depending on \x{type}}
    \While{true} \label{L:reserved-while}
    	\State \Linecomment{repeatedly try to logically delete and reserve a node}
		\State Node* $\x{last} := \x{LastDeleted}[\x{type}]$ \label{L:readLastDel}   
        \If {$\x{last} \rightarrow \nt = \x{nil}$}\label{L:readLastNext}\Comment{if list is empty}
     	    \State \Return $nil$\label{L:return-nil} \Comment{return $nil$ to indicate empty DEPQ}
	    \Else
			\State $\x{last} := \op{FAO}(\&\x{last} \rightarrow \x{next}[\x{type}], 1$)\label{L:mark}\Comment{logically delete first undeleted node in list}
			\State $\x{LastDeleted}[\x{type}] := \x{last}$\Comment{advance \x{LastDeleted}}\label{L:writeLastDel}
        	\If{$\op{TestAndSet}(\x{last} \rightarrow \x{reserved}) = 0$}\label{L:reserve}\Comment{if \op{TestAndSet} successfully reserves node}
				\State\Return $\x{last}\rightarrow \x{key}$\Comment{return key of reserved node}
			\EndIf
		\EndIf
    \EndWhile\label{L:reserved-endwhile}
\EndFunction

\medskip

\Function{\op{SL-UpdateHead}}{int \x{type}} 
	\If{$\head \neq \ld$} \Comment{if head needs updating}
		\State $\head := \ld$\label{L:writeHead} \Comment{physically delete prefix}

	\EndIf
\EndFunction
}
\end{algorithmic}
\label{alg:one-level-sl}
\end{algorithm}

Pseudocode for the list-based DEPQ is provided in \Cref{alg:one-level-sl} using the data types defined in \Cref{list-type}.
We use the convention that shared variables begin with an upper-case letter,
whereas process-local variables begin with a lower-case letter.
The shared array $\x{Head}$ stores two pointers, pointing to the first node of $L_{min}$ and $L_{max}$.
The shared array \f{LastDeleted} holds two pointers to the last logically deleted element of $L_{min}$ and $L_{max}$.
We use $\MIN=0$ and $\MAX=1$ to index the array entries for $L_{min}$ and $L_{max}$.
Each node stores a key and a \x{reserved} bit to coordinate \op{ExtractMin} and \op{ExtractMax}
operations, as described in \Cref{generic}.
The \f{next}[\MIN] and \f{next}[\MAX] fields of each node point 
to the next node in $L_{min}$ and $L_{max}$, respectively.
We assume that the least significant bit of each \x{next} field can be used
to indicate whether the pointer is \emph{marked}.  
Recall that this indicates that the node pointed to is logically deleted.
In the pseudocode,
we use the notation $\langle p, m\rangle$ to describe the pointer and the mark bit 
stored together in a \x{next} field, since they are both accessed atomically.

To simplify the code, we initialize the data structure with one logically deleted node to ensure that there is always a logically deleted node at the beginning of each list.  
Initially, \x{Head}[\MIN] and \x{Head}[\MAX] both point to this 
dummy node, whose \x{next}[\MIN] and \x{next}[\MAX] pointers are both \x{nil}.

\boldop{Insert} creates a new node with key $k$
and inserts it into the DEPQ by inserting it into $L_{min}$ first (line~\ref{L:insMin})
and then into $L_{max}$ (line~\ref{L:insMax}).
The \op{SL-Insert} function performs the insertion into one of the two singly-linked lists.
The parameter $\x{type}$ determines whether it inserts into $L_{min}$ or $L_{max}$.
To accomplish the insertion, \op{SL-Insert} first searches for the predecessor node \x{pred} and successor node \x{succ} of the node to be inserted (lines \ref{L:readcurr1}--\ref{L:ins-endwhile}). 
This search skips marked edges and uses the function \op{Before} to test whether one key should appear before another in the list.
The result of \op{Before} depends on \x{type}, since the elements
appear in opposite orders in the two lists.
The new node is inserted into the list by a successful \CAS\ at line~\ref{L:ins0-cas}.
The \CAS\  may fail either because another \op{Insert} operation changed \x{pred}'s \x{next} pointer to insert some other node, or because an \op{Extract} operation marked \x{pred}'s \x{next} pointer. 
Therefore, a \op{SL-Insert} tries repeatedly to add the node into the list 
using the repeat loop at line \ref{L:ins0-repeat}--\ref{L:ins0-cas}
until the \CAS\ succeeds.
Because logically deleted nodes form a prefix of the list and nodes are inserted after this prefix, 
the search for the location to insert the node in each iteration can continue
from where it attempted the CAS in the previous iteration.

The DEPQ uses two instances $C_{min}$ and $C_{max}$ of the CC-Synch algorithm,
one for each end of the list.
\boldop{ExtractMin} and \boldop{ExtractMax} 
call the CC-Synch algorithm on line \ref{L:DeleteMin} or \ref{L:DeleteMax}
to perform the
deletion from the corresponding list.  
We assume a call to \op{CCSynch}($f_1, f_2$) announces a request for a
combiner to perform a function call
$f_1$ on behalf of the process, and call $f_2$ just once at the end of a batch of operations.
In our algorithm, the combiner performs an \op{SL-Extract}
for each \op{Extract} operation that it serves, 
and then performs one call to \op{SL-UpdateHead}
at the end of the batch.

\boldop{SL-Extract} logically deletes a node from $L_{min}$ or $L_{max}$ 
(depending on whether $type$ is \MIN\ or \MAX),
updates $\f{LastDeleted}[\x{type}]$ and returns the key of the deleted element. 
\op{SL-Extract}\ starts from the node \ld\ and
traverses the list (lines~\ref{L:reserved-while}--\ref{L:reserved-endwhile}) marking the next pointers of nodes (lines~\ref{L:mark}), and advancing
\x{LastDeleted}[\x{type}] until it manages to change the reserved bit 
of a node from 0 to 1 with the \op{TestAndSet} of line~\ref{L:reserve}. 
This ensures that \ld\ points to the last logically deleted node in the list.
The key of that node is returned to the combiner, which informs
the process that requested the \op{Extract} operation.
\op{SL-Extract} can update \ld\ using a simple write at line \ref{L:writeLastDel}
because CC-Synch ensures there is only one active instance of 
\op{SL-Extract} at a time.

At the end of a batch of \op{SL-Extract} operations, the combiner
calls \boldop{SL-UpdateHead} to physically delete the prefix of logically deleted nodes,
except the last one, which \ld\ points to.
It does this  
by changing \head\ to point to \ld\ at line~\ref{L:writeHead}.

It is worth noticing that the DEPQ is \emph{not} a doubly-linked list:  nodes might not be in exactly opposite orders in the two lists.
\Cref{twist} shows an example, starting with the configuration shown in part (a), where
1 is logically deleted in $L_{min}$ and 5 is logically deleted in $L_{max}$.
Suppose \op{Inserts} of  3 and 4 run concurrently with an \op{ExtractMax}.
Keys 3 and 4 are inserted into $L_{min}$ in increasing order and key 3 is inserted
into $L_{max}$.
Then the \op{ExtractMax} logically deletes~3 from $L_{max}$ by  marking the
pointer to it.
The result is shown in part (b) of the figure.
Finally, key 4 is inserted into $L_{max}$ between the
logically deleted prefix (nodes 5 and 3) and node 2.  The result is shown in part (c).
Interestingly, the DEPQ is linearizable despite the lack of consistency between the two singly-linked lists.
Trying to maintain a consistent doubly-linked list would require more
complex synchronization~\cite{Sha15,ST08}.

\begin{figure}
\begin{center}
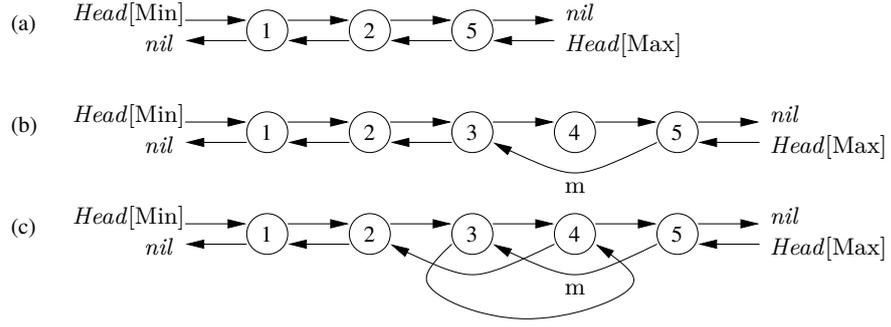
\end{center}
\vspace*{-4mm}
\caption{An example of the list-based DEPQ where orders in the two lists are different.  Part (a) shows a state where 1 is logically deleted from $L_{min}$ and 5 is logically deleted from $L_{max}$.
Then, (b) shows the result of inserting 3 and 4 into $L_{min}$ and inserting 3 into $L_{max}$, followed by the logical deletion of 3 from $L_{max}$.
Then, (c) shows the result of inserting 4 into $L_{max}$.\label{twist}}
\end{figure}

\subsection{Linearizability}

Operations on $L_{min}$ and $L_{max}$ use disjoint fields to represent
the two priority queues \x{MinPQ} and \x{MaxPQ}.
So we can focus on one of the lists to show that it correctly implements
a linearizable priority queue.  The linearizability of the DEPQ follows
from the argument in \Cref{generic}.

We say that a node becomes \emph{logically deleted} when 
a pointer to the node is marked by line \ref{L:mark}.
The dummy node initially in the list is considered logically deleted right
from the initial configuration.

\begin{invariant}\label{head-deleted}
\head\ and \ld\ point to logically deleted nodes.
\end{invariant}
\begin{proof}
This is true initially since both \head\ and \ld\ point to the one node in the list.
When \ld\ is updated on line \ref{L:writeLastDel},
the previous line ensures that the node it points to has been logically deleted.
Whenever \head\ is updated (on line \ref{L:writeHead}), the value
of \ld\ is copied into it, so this step preserves the invariant.
\end{proof}

The following lemma guarantees that insertions are performed at the correct
location in the list.

\begin{lemma}\label{pred-curr}
Consider an invocation of \op{SL-Insert}($k,\x{type}$).
Let $pr$ and $cr$ be the values of its variables $pred$ and $curr$ 
when the loop at line \ref{L:ins-while} terminates.
At the last preceding execution of line \ref{L:readcurr1} or \ref{L:readcurr2}, the following hold.
\begin{enumerate}
\item \label{pr_key} $\op{Before}(\x{pr}\rightarrow \f{key}, k, type)$ or \x{pr} is logically deleted,  
\item \label{cr_key} $\op{Before}(k,\x{cr}\rightarrow \f{key})$,  
\item \label{pr_cr} $\x{pr}\rightarrow \nt = cr$, and
\item \label{cr_unmarked} $\x{cr}\rightarrow \nt$ is unmarked.
\end{enumerate} 
\end{lemma}
\begin{proof}
We argue that \emph{every} time line \ref{L:readcurr1} or \ref{L:readcurr2}
is executed, either $\op{Before}(\x{pred}\rightarrow \f{key}, k, type)$ or \x{pred} is logically deleted.
We prove this by induction on the number of time line \ref{L:readcurr1} or \ref{L:readcurr2} has been executed by \op{SL-Insert}.
The first time an \op{SL-Insert} executes line \ref{L:readcurr1}, \x{pred} was set to the \x{Head} node of the list on line \ref{init-pred}, so that node is logically deleted by \Cref{head-deleted}.
Consider any subsequent execution of line \ref{L:readcurr1} or \ref{L:readcurr2}
and assume, as the induction hypothesis, that the claim holds for previous executions of the lines.
If line \ref{L:readcurr1} is executed, the value of \x{pred} has not changed since
the last previous execution of line \ref{L:readcurr1} or \ref{L:readcurr2},
so the claim holds by the induction hypothesis.
Line \ref{L:readcurr2} is executed only if the test on line \ref{L:ins-while}
was true, so the claim follows.
Thus, Claim \ref{pr_key} holds.

When the the loop at line \ref{L:ins-while} exits, $\op{Before}(\x{cr}\rightarrow \f{key},k,\x{type})$ is false.  Since we assume all inserted keys are unique, $\op{Before}(k,\x{cr}\rightarrow \f{key},\x{type})$ is true.
Claim \ref{pr_cr} is true because \x{cr} is read from $\x{pr}\rightarrow \nt$ at the last preceding execution of line \ref{L:readcurr1} or \ref{L:readcurr2}.
Claim \ref{cr_unmarked} follows from the exit condition of the loop at line \ref{L:ins-while}; marked edges are never changed, so if $\x{cr}\rightarrow\nt$ is unmarked at line \ref{L:ins-while}, it was unmarked at all times before that.
\end{proof}

We say a node is \emph{in the list $L_{type}$} if it can be reached by following 
\nt\ pointers from the node \head\ points to.
Let $L_{type}^{del}$ be the set of logically deleted nodes in $L_{type}$.
The following invariant, illustrated in \Cref{list-inv-fig}, describes the structure of
the list.

\begin{figure}[t]
\begin{center}
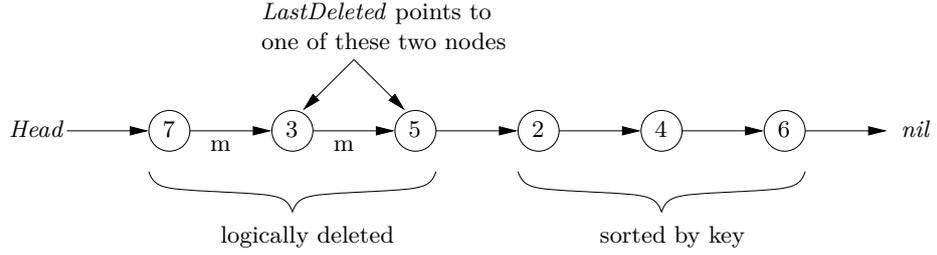
\end{center}
\vspace*{-4mm}
\caption{Structure of the list described in \Cref{big-inv}.  Marked pointers are indicated with the letter $m$.\label{list-inv-fig}}
\end{figure}

\begin{invariant}\label{big-inv}
The following claims hold in every configuration.
\begin{enumerate}
\item\label{well-formed}
The \nt\ pointers starting from \head\ form a finite linked list.
\item\label{prefix}
$L_{type}^{del}$ is a prefix of the list $L_{type}$.
\item\label{sorted}
The suffix of the list $L_{type}$ after the prefix $L_{type}^{del}$ is sorted by the $key$ field.
\item\label{prefix-marked}
\ld\ is either the second-last node of the prefix $L_{type}^{del}$ if an instance
of \op{SL-Extract}(\x{type}) has executed the FAO  of line~\ref{L:mark}, 
but it has not yet executed the write to \x{LastDeleted}[{\x type}] on line~\ref{L:writeLastDel}; or the last node of the prefix $L_{type}^{del}$ otherwise.
\item\label{last-in-list}
In any instance of \op{SL-Extract}$(\x{type})$, the variable \x{last} points to a node in  $L_{type}^{del}$.
\end{enumerate}
\end{invariant}
\begin{proof}
In the initial configuration, the \nt\ pointers starting from \head\ form a linked list of a single node.
That node is logically deleted and \ld\ points to it.
No instances of \op{SL-Extract} are active.
So all claims hold.

We show all claims are preserved by each step $s$.
Assume they hold prior to $s$.

Line~\ref{L:ins0-newSucc} updates the \nt\ field of a node that was 
created on line~\ref{L:newNode} and has never previously been connected in $L_{type}$. 
Thus, if $s$ is the execution of line~\ref{L:ins0-newSucc}, all claims still hold after $s$. 

Assume now that $s$ is an execution of a successful \op{CAS} on line~\ref{L:ins0-cas} by an insert operation \op{op}.
It changes $\x{pred}\rightarrow \nt$ from \x{curr} to the new node \x{node},
whose \nt\ field was set to \x{curr} on line \ref{L:ins0-newSucc}.
Thus, this step has the effect of inserting \x{node} into the linked
list $L_{type}$ between \x{pred} and \x{curr}, which preserves claim \ref{well-formed}.
Since the \op{CAS} succeeds only if the \nt\ pointer being changed was unmarked, the insertion does not occur in the prefix of logically deleted nodes.
Thus, claims \ref{prefix}, \ref{prefix-marked} and \ref{last-in-list} are not affected by this step.
It remains to prove claim \ref{sorted}.
By \Cref{pred-curr}.\ref{cr_key} \x{node}'s key is before \x{curr}'s key.
By \Cref{pred-curr}.\ref{pr_key} 
either \x{pred}'s key is before \x{node}'s key or \x{pred} is in $L_{type}^{del}$.
In the first case, \x{node} is inserted between \x{pred} and \x{curr}
in the non-deleted suffix of $L_{type}$, as argued above, and claim \ref{sorted} is preserved.
In the second case, \x{node} is inserted as the first node in the non-deleted suffix of $L_{type}$, preserving claim \ref{sorted}, since \x{node}'s key is before \x{curr}'s key.

If step $s$ is an execution of line \ref{L:readLastDel}, it copies
\ld\ into the local variable $last$ of an instance of \op{CombineDelete},
so claim \ref{last-in-list} holds since claim \ref{prefix-marked} is true before $s$.
Similarly, if $s$ is an execution of line \ref{L:writeLastDel},
it copies \x{last} into \ld, so the invariant is also preserved.

Suppose $s$ is an execution of line~\ref{L:mark} of \op{CombineDelete}(\x{type}).
Since combining guarantees that no other \op{CombineDelete}(\x{type}) is active,
claim \ref{prefix-marked} of the induction hypothesis ensures that \ld\ 
points to the last logically deleted node in $L_{type}$.
Then, the \op{FAO} of $s$ simply marks the \nt\ pointer of this last node, extending 
the prefix of logically deleted nodes to one more node.
After $s$, \ld\ points to the second-last node in this prefix, as required
for claim \ref{prefix-marked}.
Step $s$ also updates  \x{last} to point to that
next node, preserving claim \ref{last-in-list}.
So, all claims are preserved.

Assume  now that $s$ is an execution of line~\ref{L:writeHead}, which updates
the head pointer to \ld. 
Since \ld\ was pointing to a reachable node in $L_{type}$ prior to $s$,
this step has the effect of removing a prefix of the list.
This preserves all of the claims.
\end{proof}

\begin{lemma}\label{unmarked-reachable}
Suppose a node is inserted into $type$-list before configuration $C$, but it
is not in the list $L_{type}$ at configuration $C$.  Then, the node's \nt\ field was marked by line \ref{L:mark} of a call to \op{SL-Extract}$(\x{type})$ prior to $C$.
\end{lemma}
\begin{proof}
A node can be disconnected from the list $L_{type}$
only by line \ref{L:writeHead}, which updates \head\ to point to \ld.
The claim follows from \Cref{big-inv}.\ref{prefix-marked}.
\end{proof}

By \Cref{big-inv},
the nodes of $L_{type}$ that have not be logically deleted form a
suffix of the singly-linked list data structure, starting with the successor
of \ld, and this suffix is sorted by $key$ field.
\op{SL-Insert} adds a key to this suffix.
To extract the minimum element of $L_{type}$, \op{SL-Extract}
repeatedly tries to logically delete the first node in this suffix
at line \ref{L:mark}.
If we were simply implementing an \op{ExtractMin} operation in a
single-ended priority queue, the key
in that node would be returned.  This would ensure correctness
of the priority queue, since it would maintain the invariant that 
the keys in the suffix would always be equal
to the contents of the priority queue.
Since we are using an \op{ExtractMin} on the single-ended priority queue
implemented by $L_{type}$ inside the DEPQ construction of 
\Cref{generic-alg}, the next step would be to try set the \x{reserved}
bit of this node, and retry an \op{ExtractMin} on $L_{type}$ if that marking fails.
In \Cref{list-code}, we incorporate this attempt to set the \x{reserved}
bit into the \op{SL-Extract} operation, to skip any logically deleted
nodes of $L_{type}$ that have already been claimed by an \op{Extract} operation
at the other end of the DEPQ.
This is a straightforward simplification of the code.
Thus, it follows from the linearizability of \Cref{generic-alg} that 
the DEPQ implementation in \Cref{list-code} is linearizable
using the following linearization points.
\begin{itemize}
\item%
Each \op{Extract} operation
that receives a response of 0 from the \op{TestAndSet} on line \ref{L:reserve} is linearized
at its last execution of line \ref{L:mark}.
\item%
Each \op{Extract} operation that returns \x{nil} at line \ref{L:return-nil}
is linearized when the \nt\ pointer of $last$ is read on the preceding line \ref{L:readLastNext}.
\item%
Now consider an \op{Insert} operation.  It executes two instances $I_1$ and $I_2$ of SL-Insert.  Let
$\x{cs}_1$ and $cs_2$ be the successful executions
of the CAS of line~\ref{L:ins0-cas} by $I_1$ and $I_2$, respectively (if they exist).
The \op{Insert} is linearized at the earlier of
$cs_2$, or
immediately before the \op{Extract} operation on DEPQ that returns its item.
If neither of these events occur, then the \op{Insert} is not linearized.
\end{itemize}

\begin{theorem}
\Cref{list-code} is a linearizable multi-consumer DEPQ.
\end{theorem}

\section{Discussion}

Our DEPQ in \Cref{example}
uses the standard technique of marking pointers to logically delete nodes.
Some concurrent data structures mark nodes instead of pointers (e.g., \cite{EFRvB10}).
When applying our construction to such a data structure, it might be possible
to use just a single bit to both mark a node as logically deleted from the priority queues and to reserve the node in the DEPQ.  This would reduce the number
of RMW instructions needed for an \op{Extract} operation.

If the priority queue data structures used in our general
construction have garbage collection, then it will work in exactly the same way for the DEPQ.
However, if nodes are shared between the two original data structures
as in the example of \Cref{example},  additional care is required,
since when a process physically deletes a node from one priority queue
it does not know whether the node has also been physically deleted from the other.
Thus, it cannot tell if it is safe to retire this node (i.e., declare
that it is a node that can be deallocated when no other process holds a reference
to it).
One way to handle this issue is to store an additional bit in each node.  When a node is 
physically deleted from the first data structure, a \op{TestAndSet} sets the bit to 1.
When it is physically deleted from the second, the \op{TestAndSet} on this bit will fail,
indicating it is safe to retire the node.
Most existing garbage collection schemes (for example, epoch-based reclamation \cite{Fra04} or hazard pointers \cite{HL+05,M04}) would handle shared nodes with this simple mechanism.
This additional mechanism may not even be required for some garbage-collection schemes,
for example those based on reference counting \cite{DM+01}.

Our construction in \Cref{generic} produces a dual-consumer DEPQ.
Could a more sophisticated synchronization mechanism yield a linearizable
multi-consumer DEPQ directly from two linearizable multi-consumer priority queues?
We introduced a concurrent version of the dual queue method described
in \Cref{related-seq}.
Could  other methods discussed there also be adapted for concurrent DEPQs?

More generally, other sequential data structures can be produced by gluing together two data structures, like 
the threaded BST mentioned in \Cref{introduction}.
Is it possible to take black-box linearizable implementations
of two data structures and combine them in a way that is similar
to our construction, so that the
resulting data structure is linearizable and provides more functionality
than either of its individual components?
This might yield simpler designs than trying to maintain consistency
among all parts by hand (as is done in \cite{CNT14} for a rather
complicated concurrent threaded BST).

\bibliographystyle{plainurl}%
\bibliography{main}

\newpage

\end{document}